\documentclass{article}

\usepackage{lineno,hyperref}

\usepackage{fullpage}
\usepackage{amssymb,amsmath,hyperref,graphicx,array,caption,subcaption}
\usepackage[usenames,dvipsnames]{color}
\usepackage[ruled,noline,norelsize]{algorithm2e}
\usepackage{amsthm}

\newtheorem{theorem}{Theorem}[section]
\newtheorem{lemma}[theorem]{Lemma}
\newtheorem{corollary}[theorem]{Corollary}

\theoremstyle{definition}
\newtheorem{definition}[theorem]{Definition}

\newcommand{\eee}{\mathrm{e}}

\newcommand{\ie}{\emph{i.e.}}

\newcommand{\whp}{with high probability}

\newcommand{\concat}{\odot}

\newcommand{\prefix}{\preccurlyeq}

\newcommand{\initialRoundSize}{R_0}
\newcommand{\trans}{\mathcal{T}}
\newcommand{\AlicesTmpTrans}{\mathcal{T}_a}
\newcommand{\BobsTmpTrans}{\mathcal{T}_b}
\newcommand{\AlicesVerTrans}{\mathcal{T}^*_a}
\newcommand{\BobsVerTrans}{\mathcal{T}^*_b}

\newcommand{\fplen}{F}

\newcommand{\ec}{m}

\newcommand{\aec}{m_a}
\newcommand{\bec}{m_b}
\newcommand{\aecfinal}{m_a^*}
\newcommand{\becfinal}{m_b^*}
\newcommand{\becA}{m_b'}

\newcommand{\rs}{r}
\newcommand{\AlicesRoundSize}{r_a}
\newcommand{\BobsRoundSize}{r_b}
\newcommand{\Bobsfp}{\mathcal{F}_b}
\newcommand{\fpFromBob}{\mathcal{F}'_b}
\newcommand{\AlicesMsg}{\mathcal{F}_a}
\newcommand{\msgFromAlice}{\mathcal{F}'_a}
\newcommand{\msg}{\mathcal{F}}
\newcommand{\len}{\ell}
\newcommand{\AliceRound}{\mathcal{R}}
\newcommand{\bdAlgEnd}{12L}
\DeclareMathOperator{\isCodeword}{IsCodeword}
\DeclareMathOperator{\eAMD}{amdEnc}
\DeclareMathOperator{\dAMD}{amdDec}

\DeclareMathOperator{\eECC}{ecEnc}
\DeclareMathOperator{\dECC}{ecDec}

\DeclareMathOperator{\hash}{h}

\DeclareMathOperator{\matchesFP}{MatchesFP}

\newcommand{\nullTrans}{null}

\newcommand{\rep}{\rho}
\newcommand{\nr}{N}
\newcommand{\inr}{N_1}
\newcommand{\nrj}{N_j}
\newcommand{\twocplus}{\alpha}

\newcommand{\Cost}{Cost}

\title{Interactive Communication with Unknown Noise Rate}

\author{
Varsha Dani\thanks{Department of Computer Science, University of New Mexico. 
{\tt \{varsha, hayes, movahedi, saia\} @cs.unm.edu}}
\and Thomas P. Hayes\footnotemark[1]
\and Mahnush Movahedi\footnotemark[1]
\and Jared Saia\footnotemark[1]
\and Maxwell Young\thanks{Computer Science and Engineering Department, Mississippi State University.
{\tt myoung@cse.msstate.edu} 
\newline This research was supported in part by
NSF grants CNS-1318294 and 
CCF-1420911. 
}}

\begin{document}
 
\date{}

\maketitle

%
%
%
%

\begin{abstract}
Alice and Bob want to run a protocol over a noisy channel, where a certain number of bits are flipped adversarially. Several results take a protocol requiring $L$ bits of noise-free communication and make it robust over such a channel. In a recent breakthrough result, Haeupler described an algorithm that sends a number of bits that is conjectured to be near optimal in such a model.  However, his algorithm critically requires {\it a priori} knowledge of the number of bits that will be flipped by the adversary.

\hspace{10pt} We describe an algorithm requiring no such knowledge. If an adversary flips $T$ bits, our algorithm  sends $L + O\left(\sqrt{L(T+1)\log L} + T\right)$ bits in expectation and succeeds with high probability in $L$.  It does so without any {\it a priori} knowledge of $T$.  Assuming a conjectured lower bound by Haeupler, our result is optimal up to logarithmic factors.

\hspace{10pt}Our algorithm critically relies on the assumption of a private channel.  We show that privacy is necessary when the amount of noise is unknown. 
\end{abstract}





\section{Introduction}\label{section:intro}

How can two parties run a protocol over a noisy channel?  Interactive communication seeks to solve this problem while minimizing the total number of bits sent. Recently, Haeupler~\cite{haeupler2014interactive} gave an algorithm for this problem that is conjectured to be optimal.  However, as in previous work~\cite{schulman:communication,brakerski:fast, brakerski:efficient, braverman:towards, braverman:towards-deterministic, gelles:efficient, ghaffari:optimal, ghaffari:optimal2}, his algorithm critically relies on the assumption that the algorithm knows the noise rate in advance, \ie,\ the algorithm knows in advance the number of bits that will be flipped by the adversary.

In this paper, we remove this assumption.  To do so, we add a new assumption of privacy.  In particular, in our model, an adversary can flip an unknown number of bits, at arbitrary times, but he never learns the value of any bits sent over the channel.  This assumption is necessary: with a public channel and unknown noise rate, the adversary can run a man-in-the-middle attack to mislead either party (see Theorem~\ref{t:privateIsNecessary}, Section~\ref{sec:remarks}).


\paragraph{Problem Overview} We assume that Alice and Bob are connected by a noisy binary channel.  Our goal is to build an algorithm that takes as input some distributed protocol $\pi$ that works over a noise-free channel and outputs a distributed protocol $\pi'$ that works over the noisy channel.

We assume an adversary chooses $\pi$, and which bits to flip in the noisy channel.  The adversary knows our algorithm for transforming $\pi$ to $\pi'$.  However, he neither knows the private random bits of Alice and Bob, nor the bits sent over the channel, except when it is possible to infer these from knowledge of $\pi$ and our algorithm.

We let $T$ be the number of bits flipped by the adversary, and $L$ be the length of $\pi$.  As in previous work, we assume that Alice and Bob know $L$.


\paragraph{Our Results} Our main result is summarized in the following theorem. 

\begin{theorem}\label{thm:main}
Algorithm~3 tolerates an unknown number of adversarial errors, $T$, succeeds with high probability in the transcript length\footnote{Specifically with probability at least $1 - \frac{1}{L \log L}$}, $L$, and if successful, sends in expectation $L + O\left(\sqrt{L(T+1)\log L} + T \right)$ bits.
\end{theorem}

The number of bits sent by our algorithm is within logarithmic factors of optimal, assuming a conjecture from~\cite{haeupler2014interactive} (see Theorem~\ref{thm:Lprime}).

Results in this paper first appeared in conference proceedings~\cite{ICALP15}.



\subsection{Related Work}

For $L$ bits to be transmitted from Alice to Bob, Shannon~\cite{shannon:mathematical} proposes an error correcting code of size $O(L)$ that yields correct communication over a {\it noisy} channel with probability $1-e^{-\Omega(L)}$. At first glance, this may appear to solve our problem.  
But consider an {\it interactive} protocol with communication complexity $L$, where Alice sends one bit, then Bob sends back one bit, and so forth where the value of each bit sent {\it depends on the previous bits received}. Two problems arise. First, using block codewords is not efficient; to achieve a small error probability, ``dummy'' bits may be added to each bit prior to encoding, but this results in a superlinear blowup in overhead.  Second, due to the interactivity, an error that occurs in the past can ruin all computation that comes after it. Thus, error correcting codes fall short when dealing with interactive protocols.\smallskip

The seminal work of Schulman~\cite{schulman:deterministic,schulman:communication} overcame these obstacles by describing a deterministic method for simulating interactive protocols on noisy channels with only a constant-factor increase in the total communication complexity. This work  spurred vigorous interest in the area (see~\cite{braverman:coding} for an excellent survey). 

Schulman's scheme tolerates an adversarial noise rate of $1/240$.  It critically depends on the notion of a {\it tree code} for which an exponential-time construction was originally provided. This exponential construction time motivated work on  more efficient constructions~\cite{braverman:towards-deterministic,peczarski:improvement,moore:tree}.  There were also efforts to create alternative codes~\cite{gelles:efficient,ostrovsky:error}.  Recently, elegant computationally-efficient schemes that tolerate a constant adversarial noise rate have been demonstrated~\cite{brakerski:efficient,ghaffari:optimal2}. Additionally, a large number of powerful results have improved the tolerable adversarial noise rate~\cite{brakerski:fast,braverman:towards,ghaffari:optimal,franklin:optimal,braverman:list}.

The closest prior work to ours is that of Haeupler~\cite{haeupler2014interactive}. His work assumes a fixed and known adversarial noise rate $\epsilon$, the fraction of bits flipped by the adversary. Communication efficiency is measured by \textit{communication rate} which is $L$ divided by the total number of bits sent.  Haeupler~\cite{haeupler2014interactive} describes an algorithm that achieves a communication rate of $1 - O(\sqrt{\epsilon \log\log(1/\epsilon)}$, which he conjectures to be optimal. We compare our work to his in Section~\ref{sec:remarks}.

Feinerman, Haeupler and Korman~\cite{feinerman2014breathe} recently studied the interesting related problem of spreading a single-bit rumor in a noisy network. In their framework, in each synchronous round, each agent can deliver a single bit to a random anonymous agent. This bit is flipped independently at random with probability $1/2-\epsilon$ for some fixed $\epsilon >0$.  Their  algorithm  ensures with high probability that in $O(\log n/\epsilon^2)$ rounds and with $O(n\log n/\epsilon^2))$ messages, all nodes learn the correct rumor. They also present a majority-consensus algorithm with the same resource costs, and prove these resource costs are optimal for both problems.

\subsection{Formal Model} 

Our algorithm takes as input a protocol $\pi$ which is a sequence of $L$ bits, each of which is transmitted either from Alice to Bob or from Bob to Alice.  As in previous work, we also assume that Alice and Bob both know $L$.  We let Alice be the party who sends the first bit in $\pi$.  

\paragraph{Channel Steps}  We assume communication over the channel is synchronous and individual computation is instantaneous.  We define a \textit{channel step} as the amount of time that it takes to send one bit over the channel.  

\paragraph{Silence on the Channel} When neither Alice nor Bob sends in a channel step, we say that the channel is silent.  In any contiguous sequence of silent channel steps, the bit received on the channel in the first step is set by the adversary for free.  By default, the bit received in subsequent steps of the sequence remains the same, unless the adversary pays for one bit flip in order to change it.   In short, the adversary pays a cost of one bit flip each time it wants to change the value of the bit received in any contiguous sequence of silent steps.

\subsection{Overview of Our Result}

\paragraph{Challenges}  Can we adapt prior results by guessing the noise rate? Underestimation threatens correctness if the actual number of bit flips exceeds the algorithm's tolerance. Conversely, overestimation leads to sending more bits than necessary.  Thus, we need a protocol that adapts to the adversary's actions.

One idea is to adapt the amount of communication redundancy based on the number of errors detected thus far.  However, this presents a new challenge because the parties may have different views of the number of errors.  They will need to synchronize their adaptions over the noisy channel.  This is a key technical challenge to achieving our result.
 
Another technical challenge is termination.  The length of the simulated protocol is necessarily unknown, so the parties will likely not terminate at the same time.  After one party has terminated, it is a challenge for the other party to detect this fact based on bits received over the noisy channel. 

A high-level overview of how we address these challenges is given in Section~\ref{s:alg-overview}.

\subsection{Paper Organization}

The rest of this paper is organized as follows.  In Section~\ref{sec:alg-bounded}, we describe a simple algorithm for interactive communication that works when $T = O(L/\log L)$.  We analyze this algorithm in Section~\ref{sec:bounded-analysis}.  In Section~\ref{sec:alg-unbounded}, we describe an algorithm for interactive communication that works for any finite $T$; we prove this algorithm correction in Section~\ref{sec:analysis-unbounded}.  Section~\ref{sec:remarks} gives some relevant remarks, including justifying private channels and comparing our algorithm with past work.  Finally, we conclude and give directions for future work in Section~\ref{sec:conc}.


\section{Bounded $T$ - Algorithm}
\label{sec:alg-bounded}

In this section, we describe an algorithm that enables interactive communication problem when $T= O(L/\log L)$.


\subsection{Overview, Notation and Definitions}

\begin{figure} 
\begin{center}
\begin{tabular}{l p{15cm}}
$L$ & The length of the protocol to be simulated.\\
$\pi$ & The $L$-bit protocol to be simulated, augmented by random bits to length $
\left(1 + \left\lceil \frac{L}{\initialRoundSize} \right \rceil \right) \initialRoundSize$.  \\
$\pi[\trans, \ell]$  & The result of the computation of the next $\ell$ bits of $\pi$ after history $\trans$. \\
$\initialRoundSize$ & Initial round size in the algorithm. This is the smallest power of 2 that is greater than $\sqrt{LF}$. So $\sqrt{LF} \leq \initialRoundSize \leq 2\sqrt{LF}$ \\
$\fplen$ & The length of the fingerprint.\\
$\AlicesTmpTrans$ & Alice's tentative transcript.\\
$\BobsTmpTrans$ & Bob's tentative transcript.\\
$\AlicesVerTrans$ & Alice's verified transcript.\\
$\BobsVerTrans$ & Bob's verified transcript.\\
$\trans[0:\ell]$ & The first $\ell$ bits of $\trans$. If $|\trans| <L$ this is $\nullTrans$ \\
\end{tabular}
\caption{Glossary of Notation} \label{f:notation}
\end{center}
\end{figure}


Our algorithm is presented as Algorithm~\ref{alg:bdIC}.  The overall idea of the algorithm is simple: the parties run the original protocol $\pi$ for a certain number of steps as if there was no noise. Then, Alice determines whether an error has occurred by checking a fingerprint from Bob. Based on the result of this verification, the computation of $\pi$ either moves forward or is rewound to be performed again.  

\subsection{Helper Functions} \label{s:helper}

Before giving details of the algorithm, we first describe some helper functions and notation (see Figure~\ref{f:notation}).

\paragraph{Fingerprinting} To verify communication, we make use of the following well-known theorem.

\begin{theorem}
\label{thm:hash}~[Naor and Naor~\cite{naorandnaorJ}]
For any positive integer $\mathcal{L}$ and any probability $p$, there exists a hash function $\mathcal{F}$ that given a uniformly random bit string $S$ as the seed, maps any string of length at most $\mathcal{L}$ bits to a bit string hash value $H$, such that the collision probability of any two strings is at most $p$, and the length of $S$ and $H$ are $|S|=\Theta(\log(\mathcal{L}/p))$ and $|H| = \Theta(\log(1/p))$ bits.
\end{theorem} 

We define two functions based on this theorem, $\hash$ and $\matchesFP$.  In this section, we will write $\hash_L$ to denote that the probability of error $p$ is polynomial in $L$.  In particular, we can set $p = 1/L^{2}$, with fingerprints of size $O(\log L)$. The function $\hash_L(T)$ takes a transcript $T$ and returns a tuple $(s, f)$, where $s$ is uniformly random bit string and $f$ is the output of the hash function $\mathcal{F}$ in the theorem above when given inputs $s$ and $T$.  We refer to this tuple as the \emph{fingerprint} of $T$.  

The function $\matchesFP((s,f),T)$ takes a fingerprint $(s,f)$ and a transcript $T$.  It returns true if and only if the output of $\mathcal{F}$ when given bit string $s$ and transcript $T$ is equal to the value $f$.  In both of these functions, the total length of the fingerprint is given by the value $\fplen$, which will be defined later.

\medskip

\paragraph{Algebraic Manipulation Detection Codes}  Our result makes critical use of Algebraic Manipulation (AMD) Codes from~\cite{cramer2008detection}.  These codes provide three functions: $\eAMD$, $\isCodeword$ and $\dAMD$.  The function $\eAMD(m)$ creates an encoding of a message $m$.  The function $\isCodeword(m')$ returns true if and only if a received message $m'$ is equal to $\eAMD(m)$ for some sent message $m$.  The function $\dAMD(m')$ takes a received value $m'$, where $\isCodeword(m')$, and returns the value $m$ such that $\eAMD(m) = m'$.  Intuitively, AMD Codes enable detection of bit corruptions on encoded words, with high probability.

We make use of the following theorem about AMD codes.  This is a slight rewording of a theorem from~\cite{cramer2008detection}.

\begin{theorem}~\cite{cramer2008detection} \label{t:amd}
	For any $\delta > 0$, there exists functions $\eAMD$, $\isCodeword$ and $\dAMD$, such that, for any bit string $m$ of length $x$:
	\begin{itemize}
		\item $\eAMD(m)$ is a string of length $x+ C \log (1/\delta)$, for some constant $C$;
		\item $\isCodeword(\eAMD(m))$ and $\dAMD(\eAMD(m)) = m$;
		\item For any bit string $s \neq 0$ of length $x$, $Pr(\isCodeword(\eAMD(m) \oplus s)) \leq \delta$
	\end{itemize}	
\end{theorem}

In this section, we set $\delta = 1/L^{2}$ and add $O(\log L)$ additional bits to the message word.  Also in this section, we will always encode strings of size $O(\log L)$, so the AMD encoded messages will be of size $O(\log L)$.

In the algorithm, we will denote the fixed length of the AMD-encoded fingerprint by $\fplen$.

\subsection{Remaining Notation}

\paragraph{Transcripts} We define Alice's \emph{tentative transcript}, $\mathcal{T_A}$, as  the sequence of possible bits of $\pi$ that Alice has either sent or received up to the current time. Similarly, we let $\mathcal{T_B}$ denote Bob's transcript. For both Alice or Bob, we define a \emph{verified transcript} to be the longest prefix of a transcript for which a verified fingerprint has been received.  We denote the verified transcript for Alice as $\mathcal{T^*_A}$, and for Bob as $\mathcal{T^*_B}$. The notation $T\prefix T'$ signifies that a transcript $T$ is a prefix of a transcript $T'$.

\paragraph{Rounds} We define one of \emph{Alice's rounds} as one iteration of the repeat loop in Alice's  protocol.  Alice's round consists of $r_a$ channel steps, where $r_a$ is the \emph{round size} value maintained by Alice.  Similarly, we define one of \emph{Bob's rounds}  as one iteration of the repeat look in Bob's protocol.  Such a round consists of $r_b$ channel steps, where $r_b$ is the \emph{round size} for Bob.

\paragraph{Other Notation} For a transcript $\mathcal{T}$ and integer $i$, we define $\mathcal{T}[0 : i]$ to be the first $i$ bits of $\mathcal{T}$.  For two strings $x$ and $y$, we define $x \concat y$ to be the concatenation of $x$ and $y$.



\subsection{Algorithm Overview} \label{s:alg-overview}

\begin{algorithm*}
\caption{Bounded Error Interactive Communication}
\label{alg:bdIC}
\begin{minipage}{.43\textwidth}
\CommentSty{\bf ALICE'S PROTOCOL}
\BlankLine
\nl 
$\AlicesTmpTrans \gets \nullTrans$; $\AlicesVerTrans \gets \nullTrans$\; 
$\aec \gets 0$; $\AlicesRoundSize \gets \initialRoundSize$\; 
\nl \label{ap:repeat3}\Repeat{$\aec = \frac{\initialRoundSize^2}{4\fplen^2} -1$}{
\nl $\AlicesMsg \gets \eAMD(\aec, \AlicesRoundSize, |\AlicesVerTrans|)$\;
\nl Send $\AlicesMsg$\;
\nl Append $\pi[\AlicesTmpTrans, \AlicesRoundSize - 2\fplen]$ to $\AlicesTmpTrans$\;
\nl Receive Bob's $\fplen$-bit message, $\fpFromBob$\;
\nl \uIf{\emph{$\isCodeword(\fpFromBob)$}}{
\nl \uIf{$|\AlicesVerTrans| \ge L$}{
\nl Output $\AlicesVerTrans[0:L]$ and \\ \textbf{Terminate}\;
}
\nl $\msg \gets \dAMD(\fpFromBob)$\;
\nl \uIf {\emph{$\matchesFP(\msg,\AlicesTmpTrans)$}}{
\CommentSty{// successful round}\;
\nl $\AlicesVerTrans \gets \AlicesTmpTrans$\;
}}
\nl \uElse{
\CommentSty{// round failed }\;
\nl $\AlicesTmpTrans \gets \AlicesVerTrans$\;
\nl \label{ap:aecincr}$\aec \gets \aec +1$\;
\nl \uIf{$1+\aec$ is a power of 4}{\nl $\AlicesRoundSize \gets \AlicesRoundSize/2$\;}
}
}

\BlankLine
\BlankLine
\BlankLine
\BlankLine
\BlankLine
\BlankLine
\BlankLine
\BlankLine
\BlankLine
\BlankLine
\BlankLine
\BlankLine
\BlankLine
\BlankLine
%
\end{minipage}
\hfill 
 \begin{minipage}{.43\textwidth}
 \setcounter{AlgoLine}{0}

\BlankLine
\CommentSty{\bf BOB'S PROTOCOL}
\BlankLine
\nl 
$\BobsTmpTrans \gets \nullTrans$; $\BobsVerTrans \gets \nullTrans$\;
$\bec \gets 0$; $\BobsRoundSize \gets \initialRoundSize$\;

\nl \label{bp:repeat}\Repeat{$\bec = \frac{\initialRoundSize^2}{4\fplen^2}-1$} {
\nl Receive Alice's $\fplen$-bit message, $\msgFromAlice$\;
\nl \uIf{\emph{all bits of $\msgFromAlice$ are equal}}{
\CommentSty{// Alice has likely left}\;
\nl Output $\BobsVerTrans[0:L]$ and \\ \textbf{Terminate}\;
}
\nl \uIf{\emph{$\isCodeword(\msgFromAlice)$}}{
\nl $(\ec, \rs, \len) \gets \dAMD(\msgFromAlice)$\;
\CommentSty{// synchronize values}\;
\nl $\BobsRoundSize \gets \rs $\;
\nl $\bec \gets \ec$\;
\nl \uIf{$\len > |\BobsVerTrans|$}{
\nl $\BobsVerTrans \gets \BobsTmpTrans$\;
}
\nl \uElse{
\nl $\BobsTmpTrans \gets \BobsVerTrans$\;
}
\nl Append $\pi[\BobsTmpTrans,\BobsRoundSize - 2\fplen]$ to $\BobsTmpTrans$\;
\nl $\Bobsfp \gets \eAMD (\hash_L(\BobsTmpTrans))$\; 
\nl Send $\Bobsfp$\;
}
\nl \uElse{
\CommentSty{// corruption occurred}\;
\nl Send random bits for $\BobsRoundSize -\fplen$ steps\;
\nl $\bec \gets \bec +1$ \;
\nl \uIf{$1+\bec$ is a power of 4}{\nl $\BobsRoundSize \gets \BobsRoundSize/ 2 $\;}
}}
\end{minipage}
\end{algorithm*}


To facilitate discussion of the algorithm, we first state some important properties of rounds (proven in Section~\ref{sec:bounded-analysis}).  First, the size of any round is always a power of two.  Second, the start of each of Bob's rounds always coincides with the start of one of Alice's rounds.  This ensures that whenever Bob is listening for the message $\msgFromAlice$, Alice will be sending such a message.

We first describe one of Alice's rounds in which 1) neither Alice nor Bob terminate; and 2) there are no adversarial bit flips.  In such a round, Alice sends an encoded message containing two pieces of information. These are $\aec$, which is the number of failed rounds Alice has counted so far; and $|\AlicesVerTrans|$, which is the size of Alice's verified transcript.  

When Bob decodes this message, he synchronizes several values with Alice.  In particular, he sets his round size value, $\BobsRoundSize$, and mistake estimate value, $\bec$, so they equal the values Alice sent.  Then, based on $|\AlicesVerTrans|$, Bob either increases the length of his verified transcript, or else decreases the length of his tentative transcript.  After this synchronization, Alice and Bob both compute a certain number of bits of $\pi$ and add these to their tentative transcripts.  Finally Bob sends an encoded fingerprint to Alice.  She verifies this fingerprint, and then adds the bits of $\pi$ computed during this round to her verified transcript.

There are two key ways in which adversarial bit flips can alter the above scenario.  First, when the encoded message Alice sends containing $\aec$ and $|\AlicesVerTrans|$ is corrupted.  In this case, Bob will send random bits for the remainder of the round.  This ensures two things.  First, whenever Alice is listening for a fingerprint from Bob, Bob will either be sending a fingerprint or random bits.  Thus, \whp, the adversary will be unable to forge an encoding of a fake fingerprint by flipping bits.  Second, Bob's error count updates at the same time as Alice's.

The other key way in which adversarial bit flips can alter the ideal scenario is as follows.  The adversary flips bits in such a way that the encoded fingerprint, $\fpFromBob$ that Bob sends to Alice, fails to be a valid fingerprint for Alice's tentative transcript.  In this case, Alice rewinds her tentative transcript, increments her error count, and updates her block size.

\paragraph{Handling Termination} In previous work, since $\epsilon$ and $L'$ are known, both parties know when to terminate (or {\it leave} the protocol), and can do so at the same time. However, since we know neither parameter, termination is now more challenging.

In our algorithm,  $\pi$ is augmented with a certain number of additional bits that Alice sends to Bob. Each of these bits is set independently and uniformly at random by Alice.  Alice terminates when her verified transcript is of length greater than $L$. Bob terminates when he receives a value $\msgFromAlice$, where all bits are the same.  This conditions ensures that 1) Bob is very unlikely to terminate before Alice; and 2) Bob terminates soon after Alice, unless the adversary pays a significant cost to delay this.


\section{Bounded T - Analysis} \label{sec:bounded-analysis}

We now prove that with high probability, Algorithm~1 correctly simulates $\pi$ when $T$ is promised to be $O(L / \log L)$. Before proceeding to our proof, we define two bad events.
\begin{itemize}
\item[] \textit{Hash Collision.} Either Alice or Bob incorrectly validates a fingerprint and updates their verified transcript to include bits not in $\pi$.
\item[] \textit{Failure of AMD Codes} The adversary corrupts an encoded message into the encoding or a different message.  Or the encoding of some message, after possible adversary corruption, equals a bit string of all zeroes or all ones.  
\end{itemize} 

Throughout this section, we will assume neither event occurs.  At the end of this section, we will show that the probability that either even occurs is polynomially small in $L$.

\begin{lemma} \label{l:po2}
	Each player's round size is always a power of two.
\end{lemma}

\begin{proof}
	This is immediate from the fact that the round size starts out as a power of $2$ and the fact that each time it decreases, it decreases by a factor of $2$.
\end{proof}

\begin{lemma}\label{lem:Amonotonicity}
$\aec$ is monotonically increasing, and hence Alice's round size never increases.
\end{lemma}
\begin{proof}
This follows immediately from the fact that the only time $\aec$ changes is on Line~\ref{ap:aecincr} of Alice's protocol, when it is incremented by 1.
\end{proof}

\begin{lemma} \label{lem:bec-upper-bd}
Algorithm~\ref{alg:bdIC} has the following properties:
\begin{enumerate} 
\item When Bob starts a round, Alice starts a round,  
\item $\bec \le \aec$ at all times that Alice remains in the protocol.
\end{enumerate}
\end{lemma}
\begin{proof} 
This follows by induction on $\aec$.  

\paragraph{Base Case}
 We first show that the lemma holds while $\aec = 0$.  Note that $\bec$ can only increase after Bob has spent a round sending random bits.  During such a round, Alice will increment $m_a$ before Bob increments $m_b$.  Next, note that while $\bec = \aec = 0$, Alice and Bob both have the same round sizes, and so when Bob starts a round, Alice starts a round.
 
\paragraph{Inductive Step} 
Consider the channel step, $t$, at which Alice increases $\aec$ to some value $j>0$.  We must show that the lemma statement holds throughout the time while $\aec = j$.
By the inductive hypothesis, up to time $t$, $\bec \le \aec$, and when Bob started a round, Alice started a  round.  There are two cases for the value of $\bec$ at the end of channel step $t$.
\paragraph{Case 1} $\bec < j$.  In this case, Bob must not have received $\AlicesMsg$ at the beginning of the round he is in at channel step $t$.  Hence, Bob transmits random bits during this entire round.  Bob's round size is an integer multiple of Alice's round size (by Lemma~\ref{l:po2}).  Thus, Bob will transmit random bits throughout Alice's round begun at channel step $t+1$.  So Alice  will not receive a matching fingerprint at the end of the round she began at step $t+1$, and so she will increment $\aec$ before Bob increments $\bec$.  This will happen before Bob completes the round he is in at time $t$, so both conditions of the lemma  hold while $\aec = j$. 
\paragraph{Case 2} $\bec = j$.  Note that $\bec$ can only increase after Bob has spent a round sending random bits.  During such a round, Alice will increment $\aec$ before Bob increments $m_b$.  Thus, while $\aec = j$, $\bec = j$.  Next, note that, if $\bec = \aec = j$ at step $t$, then Alice and Bob both ended their rounds at step $t$.  Hence, during the time that $\aec =j$, when Bob starts a round, Alice starts a round.
\end{proof}

The following corollaries are immediate from the above lemma.

\begin{corollary} 
When Bob ends a round, Alice ends a round.
\end{corollary}

\begin{corollary}\label{c:bob-round}
Bob's rounds are at least as large as Alice's rounds.
\end{corollary}

The following corollary holds from the above lemma and the fact that Bob's round sizes are at least as large as Alice's.

\begin{corollary} \label{c:listening}
While both parties remain in the protocol, whenever Bob is listening for a $\AlicesMsg$, Alice is sending it.  Also, whenever Alice is listening for $\Bobsfp$, either Bob is sending it, or Bob is sending random bits.
\end{corollary}

The following lemma also follows from Lemma~\ref{lem:bec-upper-bd}.

\begin{lemma} \label{lem:aec-bec} Let $\AliceRound$ be one of Alice's rounds which starts and ends at the same time as one of Bob's rounds.  Then, at the end of $\AliceRound$, either $\aec - \bec$ is the same as it was at the beginning of $\AliceRound$ or it equals $0$ or $1$.
\end{lemma}

\begin{proof}
If $\AlicesMsg$ is corrupted at the beginning of $\AliceRound$, Bob transmits random bits for the rest of $\AliceRound$, and both Alice and Bob increment their error counts at the end, so $\aec - \bec$ stays the same.

If $\AlicesMsg$ is not corrupted at the beginning of $\AliceRound$, then Bob sets $\bec$ to $\aec$ at the beginning of $\AliceRound$, so at the end, $\aec-\bec \leq 1$.  By Lemma~\ref{lem:bec-upper-bd} (2), $\aec-\bec \geq 0$.
\end{proof}


\subsection{Phases}

We now give some definitions.

\begin{definition} \label{d:phaseRS}
	We define \emph{phase} $j$ to be all of Alice's rounds of size $\initialRoundSize / 2^{j}$.
\end{definition}

\begin{definition} \label{d:Delta}
	 We define $\Delta_j$, for all $j>0$, to be the value $\aec-\bec$ at the end of phase $j$.
\end{definition}

Note that at the beginning of phase $j$, Alice's error count is $4^j - 1$.  We now give a few lemmas about phases.

\begin{lemma} \label{l:phaseNumRounds}
	For any $j>0$, phase $j$ contains at least $3\Delta_{j-1}$ of Alice's rounds, 
\end{lemma} 

\begin{proof}
	Consider any $j>0$.  At the beginning of phase $j$, $\aec = 4^{j} - 1$.  Also, at the beginning of phase $j$, by Lemma~\ref{lem:bec-upper-bd} (2), $\bec \leq \aec$.  Hence, $0 \leq \Delta_{j-1} \leq 4^{j} - 1$.   Note that $\aec$ increases by at most $1$ in each of Alice's rounds.  Thus, $3\Delta_{j-1}$ rounds after the beginning of phase $j$, the value of $\aec$ is at most:
 
\begin{align*}
4^j -1 + 3\Delta_{j-1} &\le  4^j -1 + 3 (4^{j}-1)\\
& < 4^{j+1} -1
\end{align*}
Thus after $3\Delta_{j-1}$ rounds, $\aec$ is not large enough for Alice to advance to phase $j+1$.
\end{proof}

\paragraph{Progressive, Corrupted and Wasted Rounds}  
Let $\AliceRound$ be one of Alice's rounds. We call $\AliceRound$  \emph{progressive} if Alice does not update her error count 
during the round, or equivalently if her verified transcript length increases. We 
call  $\AliceRound$ \emph{corrupted} if the adversary flipped at least one bit in the round. We call $\AliceRound$  \emph{wasted} if it is neither progressive nor corrupted. We  want to bound the number of wasted rounds since this number represents amount by which $\aec$ is potentially an overestimate of $T$.

We note that wasted rounds occur only when $\BobsRoundSize > \AlicesRoundSize$.  In this case, Bob is not listening when Alice sends him $\AlicesMsg$. As a result, Bob does not send Alice a valid fingerprint at the end of her round, and so her verified transcript does not increase, even though the adversary has not flipped any bits.
 
\medskip

The following lemma bounds the number of wasted rounds in a phase, and gives other critical properties.

\begin{lemma}  \label{lem:aec-becW} Suppose at the beginning of phase $j$, $j>0$, Bob is at the start of a round and his round size is at most $R_0/2^{j-1}$.  Then 
\begin{enumerate}
	\item There are at most $\Delta_{j-1}$ wasted rounds in phase $j$;
	\item $\Delta_j \in \{0, 1, 2\Delta_{j-1} \}$; and
	\item Bob ends a round at the end of phase $j$.
\end{enumerate}
\end{lemma}
\begin{proof}
 
If Bob's round size initially less than $R_0/2^{j-1}$, then it must equal $R_0/2^{j}$ in order to be a power of two.  Hence Alice and Bob will have rounds that are the same size for the entire phase, and the lemma holds trivially.
 
We now consider the harder case where Bob's round size equals $R_0/2^{j-1}$.
 
By Definition~\ref{d:phaseRS}, Alice has round size $R_0/2^j$ throughout phase $j$.  By Lemma~\ref{lem:bec-upper-bd} (2), Bob's round size is always greater than or equal to Alice's round size.  Thus, as soon as 1) Bob receives $\AlicesMsg$ in one of his rounds in phase $j$, or 2) Bob sets $\bec$ equal to Alice's error count at the beginning of phase $j$, then Bob's round size will be $R_0/2^j$ for the remainder of the phase.   Finally, by Lemma~\ref{lem:bec-upper-bd} (1), from that point on, Alice and Bob will begin, and thus end, all rounds at the same time.

Now consider Bob's rounds in phase $j$.  Assume the adversary corrupts $\AlicesMsg$ in Bob's rounds $1$ through $i$ for some value $i \geq 0$, and then the adversary does not corrupt $\AlicesMsg$ in Bob's round $i+1$.  We consider two cases.  

\paragraph{Case 1: $i < \Delta_{j-1}$}

Each of the first $i$ rounds of Bob spans two rounds of Alice.  By Lemma~\ref{l:phaseNumRounds}, these rounds are all contained in phase $j$.  Consider each pair of Alice's rounds spanned by one of Bob's rounds. The first round in the pair is corrupted, but during the second, Bob is transmitting random bits and Alice will not receive a fingerprint from him.  Thus, this round is wasted.  Hence, there are $i$ wasted rounds.  

In round $i+1$, Bob synchronizes his round size with Alice since he receives $\AlicesMsg$.  Thus, there are no more wasted rounds.  Applying Lemma~\ref{lem:aec-bec} for the remaining rounds of the phase, we see that at the end of the phase, $\aec-\bec = \Delta_{j}$ is either $0$ or  $1$.

\paragraph{Case 2: $i \geq \Delta_{j-1}$}
Bob increases $\bec$ by $1$ in each of his first $i$ rounds.  Note that at the beginning of phase $j$, Alice's error count is $4^j -1$.  Thus, after Bob's first $i$ rounds, $\bec = (4^j - 1) - \Delta_{j-1} + i$.  Hence when $i = \Delta_{j-1}$, $\bec = (4^j - 1)$.  At that time, Bob sets his round size to $R_0/2^j$, and so Alice and Bob will have the same round sizes, and will hence begin and end all rounds at the same step, for the rest of phase $j$.  Thus, there are no more wasted rounds.  Note that in this case, at Bob's $\Delta_{j-1}$ round, $\aec-\bec$  will be $2\Delta_{j-1}$.  Applying Lemma~\ref{lem:aec-bec} for the remaining rounds of the phase, we see that $\Delta_j = 2\Delta_{j-1}$, or $\Delta_j$ is $0$ or $1$.  
\end{proof}

\begin{lemma}  \label{lem:phaseJ} For every $j \geq 0$:
\begin{enumerate}
	\item There are at most $2^{j-1}$ wasted rounds in phase $j$;
	\item $\Delta_j \leq 2^j$; and
	\item Bob ends a round at the end of phase $j$.
\end{enumerate}
\end{lemma}

\begin{proof}
We prove this by induction on $j$.
\paragraph{Base Case}
At the beginning of phase $0$, Bob is at the start of a round and his round size is $R_0$.  Thus, by Lemma~\ref{lem:aec-becW}: there are $0$ wasted rounds in phase $0$; $\Delta_{0} \leq 1$; and Bob ends a round at the end of phase $0$.

\paragraph{Inductive Step}  Consider some $j>0$.  By the inductive hypothesis,  $\Delta_{j-1} \leq 2^{j-1}$.  At the beginning of phase $j$, $\bec = \aec - \Delta_{j-1} \leq (4^{j}-1) -\Delta_{j-1}$, so that $\BobsRoundSize = \initialRoundSize /2^{\lfloor\log_4{(1 + \bec)}\rfloor} \leq  \initialRoundSize /2^{\lfloor\log_4{(4^j- \Delta_{j-1})}\rfloor} \leq \initialRoundSize /2^{j-1}$. The last line holds since $0 \leq \Delta_{j-1} \leq 2^{j-1}$.

Also, by the inductive hypothesis, Bob ended a round at the end of phase $j-1$, and so is starting a round at the beginning of phase $j$.  Hence, we can apply  Lemma~\ref{lem:aec-becW} to phase $j$.  From this lemma, it follows that 1) the number of wasted rounds in phase $j$ is at most $2^{j-1}$; 2) $\Delta_j \leq 2 \Delta_{j-1} \leq 2^j$; and 3) Bob ends a round at the end of phase $j$.  
\end{proof}

Note from the above lemma that Bob's rounds are never more than double the size of Alice's rounds.  The following lemma sums up what we now know about Alice and Bob's rounds.

\begin{lemma} \label{l:bround}
The following are always true.
\begin{enumerate}
	\item Bob's round size is either equal to Alice's round size or double Alice's round size. 
	\item If Bob's round size equals Alice's round size, then when Alice starts a round, Bob starts a round.
	\item If Bob's round size is twice Alice's round size, then when Alice starts a round, either Bob starts a round, or Bob is in the middle of a round.
\end{enumerate}
	\end{lemma}

\begin{proof}
	The lemma follows from Corollary~\ref{c:bob-round}, Lemma~\ref{lem:bec-upper-bd}, and Lemma~\ref{lem:phaseJ}.
\end{proof}


\subsection{Correctness and Termination}\label{sec:termination}

\begin{lemma} \label{lem:termAlice}
	It is always the case that $\AlicesVerTrans \prefix \pi$, where $\pi$ is the padded transcript.
\end{lemma}

\begin{proof}
	This holds by Lemma~\ref{lem:collision} and Lemma~\ref{lem:amdfail}  and the fact that Alice never adds any string to $\AlicesVerTrans$ that is not verified by an encoded fingerprint from Bob.
\end{proof}

\begin{lemma} \label{l:prefixes}
At the beginning and end of each of Alice's rounds, 
\[
\BobsVerTrans \prefix \AlicesVerTrans 
= \AlicesTmpTrans
\prefix \BobsTmpTrans;
\] 
where at most one of the inequalities is strict.
Moreover, at the end of a channel step in which Bob receives $\AlicesMsg$ correctly, \[ \BobsVerTrans = \BobsTmpTrans = \AlicesVerTrans. \]
\end{lemma}
\begin{proof} 

We prove this by induction on Alice's round number.

\paragraph{Base Case}

At the beginning of the algorithm, all transcripts are $\nullTrans$, so $\BobsVerTrans = \AlicesVerTrans  = \AlicesTmpTrans = \BobsTmpTrans$.  Moreover if Bob receives $\AlicesMsg$ correctly in this round, then $\BobsVerTrans = \BobsTmpTrans = \AlicesVerTrans$.

\paragraph{Inductive Step}

We must show that the lemma holds for the $j$-th round.  By the inductive hypothesis, at the end of the $j-1$-th round,
 \[
\BobsVerTrans \prefix \AlicesVerTrans = \AlicesTmpTrans
\prefix \BobsTmpTrans,
\] 
with at most one of the inequalities being strict.  Clearly the statement about the inequalities will thus hold at the beginning of the $j$-th round.  

Alice's $j$-th round starts with Alice sending Bob $\AlicesMsg$.  

\paragraph{Case 1: Bob does not receive $\AlicesMsg$}  If Bob does not receive $\AlicesMsg$, then either 1) he was listening and it was corrupted; or 2) he was not listening for it.  If he was listening and $\AlicesMsg$ was corrupted, then Bob transmits random bits for the remainder of his round, which will be the remainder of Alice's round by Lemma~\ref{l:bround}.  By the same lemma, if Bob was not listening, then he must be in the middle of a round that is twice as large as Alice's.  In either case, Bob transmits random bits for the remainder of Alice's $j$-th round. 

Thus, Alice does not receive a matching fingerprint from Bob at the end of her $j$-th round. Thus, at the end of her round, $\AlicesTmpTrans \gets \AlicesVerTrans$ and $\BobsTmpTrans$ and $\BobsVerTrans$ are unchanged.  Hence, it continues to hold that: 
\[
\BobsVerTrans \prefix \AlicesVerTrans = \AlicesTmpTrans
\prefix \BobsTmpTrans;
\] 

and at most one of the inequalities is strict.

\paragraph{Case 2: Bob receives $\AlicesMsg$}

If Bob receives $\AlicesMsg$, then he learns the length of $\AlicesVerTrans$ and also Alice's round size. By the inductive hypothesis,  either $\AlicesVerTrans =\BobsVerTrans$ or $\AlicesVerTrans =\BobsTmpTrans$. Based on the length of $\AlicesVerTrans$, Bob either updates $\BobsVerTrans$ or rewinds $\BobsTmpTrans$, so that $\BobsVerTrans = \BobsTmpTrans = \AlicesVerTrans$.  This establishes the second part of the lemma for the $j$-th round.

Next Alice and Bob continue their rounds which are the same size.  If Alice receives a correct fingerprint from Bob at the end of her round, then the following holds: 
\[
\BobsVerTrans \prefix \AlicesVerTrans = \AlicesTmpTrans = \BobsTmpTrans.
\] 

If Alice does not receive a correct fingerprint from Bob at the end of her round, then the following holds:
\[
\BobsVerTrans = \AlicesVerTrans = \AlicesTmpTrans \prefix \BobsTmpTrans.
\] 

In either case, the first part of the lemma statement holds at the end of Alice's $j$-th round.
 \end{proof}

\begin{lemma} \label{lem:termBob} 
Bob leaves after Alice.  When Alice leaves, $|\BobsVerTrans| \geq L$.
\end{lemma}

\begin{proof} 
Bob leaves only when he receives an $\msgFromAlice$ that is all zeroes or all ones. By Lemma~\ref{lem:amdfail}, $\msgFromAlice$ is never such a string, and the adversary cannot convert $\AlicesMsg$ to such a string by bit flipping. It follows that Bob receives such a string only after Alice has left.

Alice leaves only when 1) she has received an encoded fingerprint from Bob; and 2) $|\AlicesVerTrans| \geq L$.  If Alice receives a correctly encoded fingerprint from Bob, then by Lemma~\ref{lem:amdfail}, Bob must have sent one, and hence Bob must be in a round where he received $\AlicesMsg$ correctly.  By Lemma~\ref{l:prefixes}, at that channel step, $\BobsVerTrans = \BobsTmpTrans = \AlicesVerTrans$.  Hence at the step when Alice receives the encoded fingerprint from Bob, $\BobsVerTrans = \AlicesVerTrans$.  Thus, when Alice leaves, $|\BobsVerTrans| \geq L$.
\end{proof}

\begin{lemma} \label{l:correctness}
	When either party terminates, their output is correct.
\end{lemma}

\begin{proof}
	The proof follows from Lemmas~\ref{lem:termAlice},~\ref{l:prefixes}, and~\ref{lem:termBob}, and the fact that when either party terminates, they output the first $L$ bits of their verified transcript.
\end{proof}



\subsection{Cost}

\begin{lemma} \label{lem:boberrors}
After Alice leaves, the adversary must flip at least one bit for each of Bob's rounds that does not result in Bob leaving.
\end{lemma}

\begin{proof} After Alice has left, there is silence on the channel in the steps when Bob is listening for Alice's encoded message. This means that if there is no bit flipping by the adversary, the channel transmits the same bit in every channel step, causing Bob to read a string of all zeroes or all ones, and terminate. Thus, the adversary must flip  at least one bit each time Bob is listening for a codeword.
\end{proof}

\begin{lemma} \label{lem:2^j-wasted}
There are at most $2^{j}-1$ wasted rounds prior to the end of phase $j$, for all $j \geq 0$.  
\end{lemma}
\begin{proof}
This follows trivially by repeated applications of Lemma~\ref{lem:phaseJ} (1).
\end{proof}

Throughout this section, we assume the worst case, that the adversary corrupts at most one bit per corrupted round.

\begin{lemma} \label{lem:aec}
At all times, $\aec \le T + \sqrt{T}$.  In particular, there are no more than $\sqrt{T}$ wasted rounds.
\end{lemma}

\begin{proof} 

By way of contradiction, assume $m_a > T + \sqrt{T}$ at some step, in some phase $j$, $j \geq 0$.  Then the number of wasted rounds at this step must be greater than $\sqrt{T}$.  But by Lemma~\ref{lem:2^j-wasted}, the number of wasted rounds at the end of phase $j$ is no more than $2^j-1$.  Thus, we have $\sqrt{T} < 2^j - 1$, or $T < (2^j - 1)^2$.

But $m_a$ is no larger than the number of corrupted rounds plus the number of wasted rounds.  By the above paragraph, $T < (2^j - 1)^2$ and the number of wasted rounds is no more than $2^j-1$.  Thus $m_a < (2^j - 1)^2 + (2^j - 1)$.  Moreover, we know that in phase $j$, $m_a \geq 4^j - 1$.  Thus, we know
\[
 4^j - 1 < (2^j - 1)^2 + (2^j - 1).
\]
Simplifying, we get $2^j < 1$, which is a contradiction for any $j \geq 0$.
\end{proof}

Let $\aecfinal$ denote Alice's error count when she leaves the algorithm, and $\becfinal$ denote Bob's error count when he himself leaves the algorithm. 

\begin{lemma}\label{lem:acost} 
Alice terminates in at most $L+ O(\sqrt{LF(1+\aecfinal)})$ steps. 
\end{lemma}

\begin{proof}

We first calculate the cost of the rounds that are not progressive for Alice.  The number of non-progressive rounds that she has executed is $\aecfinal$.  Her cost for these rounds is at most the following.

\begin{align*}
	\sum_{i=1}^{\aecfinal} \frac{\initialRoundSize}{2^{\lfloor \log_4 i \rfloor}} & \leq  2\initialRoundSize \sum_{i=1}^{\aecfinal} \frac{1}{2^{\log_4 i}} \\
		& = 2\initialRoundSize \sum_{i=1}^{\aecfinal} \frac{1}{\sqrt{i}} \\
		& \le  2\initialRoundSize \int_{0}^{\aecfinal} \frac{1}{\sqrt{i}} \\
	& =  4 \initialRoundSize \sqrt{\aecfinal}  
\end{align*}

In every progressive round, except possibly the last, Alice's block size is at least $R_0 2^{-\log_4 (1+\aecfinal)}$.  Thus in all but possibly the last progressive round, Alice always adds bits to her verified transcript at a rate of at least

\[
\frac{R_0 2^{-\log_4 (1+\aecfinal)} - 2F}{R_0 2^{-\log_4 (1+\aecfinal)}}.
\]

Thus, the total number of bits Alices sends in all but the last progressive round is no more than 

\[
L \cdot \frac{R_0 2^{-\log_4 (1+\aecfinal)}}{R_0 2^{-\log_4 (1+\aecfinal)} - 2F}.
\]

We will make use of the inequality
\[
\frac{1}{1-\delta} \le 1+2\delta    \ \ \ \    \mbox{ for }  0 < \delta \le 1/2
\]

and let $\delta = 2F/R_0 2^{-\log_4 (1+\aecfinal)}$.  Note that $\delta \leq 1/2$, since Alice's round size is always at least $4F$.

Then we have that the total number of bits sent by Alice in all but the last progressive round is no more than 

\[
L + \frac{4LF}{R_0 2^{-\log_4 (1+\aecfinal)}}.
\]

Adding in the last progressive round, we get that the total number of bits sent by Alice in progressive rounds is no more than

\[
L + \frac{4LF}{R_0 2^{-\log_4 (1+\aecfinal)}} + R_0 2^{-\log_4 (1+\aecfinal)}.
\]

Putting this together with the number of bits send in non-progressive rounds, we have that the total number of bits send by Alice is no more than

\begin{align*}
	L + 4 \initialRoundSize \sqrt{\aecfinal} + \frac{4LF}{R_0 2^{-\log_4 (1+\aecfinal)}} + R_0 2^{-\log_4 (1+\aecfinal)} & \leq  L + 5 \initialRoundSize \sqrt{\aecfinal} + 4 \sqrt{LF} (2^{\log_4 (1+\aecfinal)}) \\
	& \leq  L + 10 \sqrt{LF\aecfinal} + 4 \sqrt{LF(1+ \aecfinal)}\\
		& \leq  L + 14 \sqrt{LF(1+ \aecfinal)} \qedhere
\end{align*}
\smallskip
\end{proof}

\begin{lemma}\label{lem:bcost}
Bob terminates in at most $L+ 14 \sqrt{LF (1+\aecfinal)} + 8\sqrt{LF\becfinal }$ steps.
\end{lemma}

\begin{proof} Since Bob never leaves before Alice, Bob's cost must be at least as much as Alice's.  We now compute Bob's additional cost. 

At the time of Alice's departure, $\AlicesRoundSize =\initialRoundSize/2^{\lfloor \log_4 (1+ \aecfinal)\rfloor} $. By Lemma~\ref{l:bround},  $\BobsRoundSize \le 2\initialRoundSize/2^{\lfloor \log_4 (1+ \aecfinal)\rfloor} $.  Let $\becA$ denote Bob's error count when Alice leaves the algorithm.  Then $1+\becA \ge 4^{\lfloor\log_4 (1+\aecfinal)\rfloor -1} $. Bob's final error count is $\becfinal$.
Thus, Bob's additional cost is at most 

\begin{align*}
	\sum_{i= \becA}^{\becfinal-1} \frac{\initialRoundSize}{2^{\lfloor \log_4(1+ i) \rfloor}} & \leq  2 \initialRoundSize \sum_{i= 1}^{\becfinal} \frac{1}{2^{\log_4 i}}\\
		& =  2 \initialRoundSize \sum_{i= 1}^{\becfinal} \frac{1}{i^2} \\
		& \leq  4 \initialRoundSize \sqrt{\becfinal} \\
		& \leq  8 \sqrt{LF\becfinal}
\end{align*}

Combining this with Alice's cost gives the result. 
\end{proof}

\begin{lemma} \label{l:numSteps}
The algorithm ends in at most $12L$ time steps.
\end{lemma}

\begin{proof}

By Lemma~\ref{lem:bcost}, Bob terminates in at most $L+ 14 \sqrt{LF (1+\aecfinal)} + 8\sqrt{LF\becfinal }$ steps.  Moreover, $\aecfinal$ and  $\becfinal$ are no more than $\initialRoundSize^2/4F^2 -1$.  Thus, the algorithm terminates in at most the following number of steps.
\begin{align*}
L + 14 \sqrt{LF(1+\aecfinal)} + 8 \sqrt{LF\becfinal }  &\le L + 22 \sqrt{\frac{LF \initialRoundSize^2}{4F^2}}\\
&= L+ 22 \sqrt{\frac{L^2}{4}}\\
& = 12L \,.   \qedhere
\end{align*}
\end{proof}

\begin{lemma}\label{lem:A1cost}
If $T \le  \frac{L}{8F} -1$ then both players terminate with the correct output in at most $L + O(\sqrt{LF(T+1)})$ steps. 
\end{lemma}

\begin{proof} 

Let $T_a$ denote the number of bits flipped by the adversary while Alice is still in the protocol, and $T_b$ the bits flipped after Alice has left. Then $T_a + T_b = T$.

By Lemma~\ref{lem:aec}, $\aecfinal \le T_a +\sqrt{T_a}$. By Lemmas~\ref{lem:bec-upper-bd} and~\ref{lem:boberrors}, $\becfinal \le \aecfinal + T_b$. Since $T_a + T_b =T$ it follows that 
\[
\aecfinal \le T+\sqrt{T} \le 2T \le \frac{L}{4F} -2 < \frac{\initialRoundSize^2}{4F^2}-1
\]
and similarly 
\[
\becfinal <  \frac{\initialRoundSize^2}{4F^2}-1
\]

Thus, Alice and Bob will both terminate by outputting the bits of $\pi$ by Lemma~\ref{l:correctness}.

Plugging $\aecfinal \leq 2T$ and $\becfinal \leq 3T$ into Lemma~\ref{lem:bcost} gives the total number of steps required.
\end{proof}

\begin{lemma} \label{lem:collision}  
With high probability in $L$, there are no hash collisions.
\end{lemma}

\begin{proof}
By Lemma~\ref{l:numSteps}, the algorithm ends in at most \bdAlgEnd~steps. Also, there are at least $4\fplen = \Theta(\log L)$ steps in a round. Thus, the algorithm has at most $O(L\log L)$ rounds. Each round has one fingerprint.   By Theorem~\ref{thm:hash} and the setting of our fingerprint sizes, each fingerprint fails with probability at most $1/L^{2}$.  Thus, a simple union bound gives the result.
\end{proof}

\begin{lemma}\label{lem:amdfail}
With high probability in $L$, any bit flipping of a AMD encoded message is detected.
\end{lemma}

\begin{proof}
We noted in the previous lemma that the algorithm terminates in $O(L\log L)$ rounds. Each round has two AMD encoded messages.  By Theorem~\ref{t:amd} and the setting of our encoding sizes, each AMD encoding fails with probability at most $1/L^{2}$.  Again, a union bound gives the result.\end{proof}


\section{Unbounded $T$ - Algorithm} \label{sec:alg-unbounded}
 Algorithm 1 uses fingerprints of a fixed size, $F$  in order to check its transcripts. Each of these has a $1/L^2$ chance to fail due to a hash collision. Since the algorithm only computes about $O(L/\log L)$ fingerprints, a union bound tells us that with high probability  the algorithm succeeds, below its threshold value of $T$. When $T$ is large, many more checks may need to be made, and eventually there will be a good chance that there is a hash collision. Since the algorithm cannot really recover from a hash collision, we cannot afford this. On the other hand, we cannot simply start out with larger fingerprints, both because this would be too expensive if $T$ turned out to be small, and also because even bigger  fingerprints are still of a fixed size and eventually become unreliable. A natural solution is is to allow the fingerprints to grow, adapting the size to the value of $T$ seen so far, and this is indeed what we will do.

\subsection{Helper Functions}

As in Algorithm 1, we make black-box use of the Naor and Naor hash family, as well as AMD codes to protect information. However, in Iteration $j$ we need the failure probabilities for both these primitives to be $1/(2^jL^2)$. Thus, we want the fingerprint size to grow with $j$. We will denote the hash function which has a collision probability of at most $1/(2^jL^2)$ by $\hash_j$. \footnote{By abuse of notation, we will \emph{not}  subscript all the other helper functions with $j$; it should be clear from context that the version of the function used is the one that operates on strings of the correct size and has the correct failure probability} It is easy to see that $O(j)$ extra bits are required for this, so that the fingerprint size is $O(j+\log L)$.

Algorithm 1 works well when the adversary can only afford to flip a fraction of a bit per block of the algorithm.  In this case, it doesn't matter that he can corrupt an entire round of the protocol by flipping a single bit.   However, when the adversary has a larger budget, it becomes crucial to force him to pay a larger price to corrupt a round.  To this end, we wrap each fingerprint and protocol bit in a linear error-correcting code. 

To be concrete, we will use a repetition code for each protocol bit, and a Reed-Solomon code~\cite{doi:10.1137/0108018} to provide the already AMD-encoded messages with a degree of error correction.  This enables us to encode a message so that it can be recovered even if the adversary corrupts a third of the bits.  We will denote the encoding and decoding functions by $\eECC$ and $\dECC$ respectively. The following theorem, a slight restatement from~\cite{doi:10.1137/0108018}, gives the properties of these functions. 

\begin{theorem}~\cite{doi:10.1137/0108018} \label{l:ecc}
There is a constant $c>0$ such that for any message $m$, $|\eECC(m)| \le c|m|$. Moreover, if $m'$ differs from $\eECC(m)$ in at most one-third of its bits, then $\dECC(m') = m$.   
\end{theorem}

Finally, we observe that the linearity of $\eECC$ and $\dECC$  ensure that when the error correction is composed with the AMD code, the resulting code has the following properties:
\begin{enumerate}
\item If  at most a third of the  bits of the message are flipped, then the original
message can be uniquely reconstructed by rounding to the nearest
codeword in the range of $\eECC$.
\item Even if an arbitrary set of bits is flipped, the probability of
the change not being recognized is at most $\delta$, $\ie$ the same guarantee as the AMD codes.
\end{enumerate}
This is because $\dECC$ is linear, so when noise $\eta$ is added by
the adversary to the codeword $x$, effectively what happens is the
decoding function $\dECC(x+\eta) = \dECC(x) + \dECC(\eta) = m + D(\eta)$, where
$m$ is the AMD-encoded message.  But now $\dECC(\eta)$ is an obliviously
selected string added to the AMD-encoded codeword.

\subsection{Algorithm}

\begin{algorithm*}
\caption{Interactive Communication: Iteration $j$}
\label{alg:IC}
\begin{minipage}{.43\textwidth}
\CommentSty{\bf ALICE'S PROTOCOL}
\BlankLine

\CommentSty{Parameters: $\nrj, \fplen_j, \rep_j$}\;
\nl \label{ap:repeat1}\For{ $i = 1$ to $\nrj$}{
\nl $\AlicesMsg \gets \eECC(\eAMD( |\AlicesVerTrans|))$\;
\nl Send $\AlicesMsg$\;
\nl \uIf{$|\AlicesVerTrans| < L$}{
\nl \For{the next $\lfloor\fplen_j/\rep_j\rfloor$ bits of $\pi$}{
\nl \uIf{sender}{
\nl Send next bit $\rep_j$ times\; 
\nl Append to $\AlicesTmpTrans$\;
} 
\nl \uElse{ 
\nl Receive $\rep_j$ bits\;
\nl Append majority bit to $\AlicesTmpTrans$\;
}}}
\nl \uElse{
\nl Transmit $\fplen_j$ random bits.
}
\nl Receive Bob's $c \fplen_j$-bit message, $\fpFromBob$\;
\nl \uIf{\emph{$\isCodeword(\fpFromBob)$}}{
\nl \uIf{$|\AlicesVerTrans| \ge L$}{
\nl Output $\AlicesVerTrans[0:L]$ and \\ \textbf{Terminate}\;
}
\nl $\msg \gets \dAMD(\fpFromBob)$\;
\nl \uIf {\emph{$\matchesFP(\msg,\AlicesTmpTrans)$}}{
\CommentSty{// successful round}\;
\nl $\AlicesVerTrans \gets \AlicesTmpTrans$\;
}}
\nl \uElse{
\CommentSty{// round failed }\;
\nl $\AlicesTmpTrans \gets \AlicesVerTrans$\;
}
}

\BlankLine
\BlankLine
\BlankLine
\BlankLine
\BlankLine
\BlankLine
\BlankLine
\BlankLine
\BlankLine
\BlankLine
\BlankLine
\BlankLine
\BlankLine
\BlankLine
%
\end{minipage}
\hfill 
 \begin{minipage}{.43\textwidth}
 \setcounter{AlgoLine}{0}

\BlankLine
\CommentSty{\bf BOB'S PROTOCOL}
\BlankLine

\CommentSty{Parameters: $\nrj, \fplen_j, \rep_j$}\;
\nl \label{ap:repeat2}\For{ $i = 1$ to $\nrj$}{
\nl \uIf{$|\BobsVerTrans| \ge L$}{
\nl Wait $c\fplen_j$ channel steps\;
\nl Receive $\fplen_j$ bits\;
\nl \uIf {fewer than $\fplen_j/3$ alternations in the received string}{
\nl Output $\BobsVerTrans[0:L]$ and \\ \textbf{Terminate}\;
} 
\nl \uElse{
\nl $\Bobsfp \gets \eECC(\eAMD (\hash_j(\BobsVerTrans)))$\; 
\nl Send $\Bobsfp$\;}
}
\nl \uElse{
\nl Receive Alice's $c\fplen_j$-bit message $\msgFromAlice$\;
\nl \uIf{\emph{$\isCodeword(\dECC(\msgFromAlice))$}}{
\nl $\len \gets \dAMD (\dECC(\msgFromAlice))$\;
\nl \uIf{$\len > |\BobsVerTrans|$}{
\nl $\BobsVerTrans \gets \BobsTmpTrans$\;
}
\nl \uElse{
\nl $\BobsTmpTrans \gets \BobsVerTrans$\;
}
\nl \For{the next $\lfloor\fplen_j/\rep_j\rfloor$ bits of $\pi$}{
\nl \uIf{sender}{
\nl Send next bit $\rep_j$ times\; 
\nl Append to $\BobsTmpTrans$\;
} 
\nl \uElse{ 
\nl Receive $\rep_j$ bits\;
\nl Append majority bit to $\BobsTmpTrans$\;
}}
\nl $\Bobsfp \gets \eECC(\eAMD (\hash_j(\BobsTmpTrans)))$\; 
\nl Send $\Bobsfp$\;
}
\nl \uElse{
\nl Transmit $(c+1)\fplen_j$ random bits.
}}}


\end{minipage}
\end{algorithm*}



\begin{algorithm*}
\caption{Interactive Communication}
\label{alg:IC2}
\begin{minipage}{.43\textwidth}
\CommentSty{\bf ALICE'S PROTOCOL}
\BlankLine
\CommentSty{\bf // Iteration 0}\;
\nl Run Alice's protocol from Alg 1 \;
\nl \uIf{not terminated}{
\nl transmit random bits until channel step $\bdAlgEnd$\;
}
\CommentSty{\bf // End of Iteration 0}\;
\nl $j \gets 1$\;
\nl \While{still present}{ 
\CommentSty{\bf // Iteration $j$}\;
\nl $\fplen_j \gets \beta(j + \log L)$\;
\nl $\rep_j \gets 2^j \lceil\frac{\fplen_j}{\fplen}\rceil \wedge \fplen_j$\;
\nl $\nrj \gets 2^{j-1} \lceil 8L/\fplen \rceil$\;
\nl Run Alice's protocol from Algorithm 2, with parameters $\nrj, \fplen_j, \rep_j$\;
\CommentSty{\bf // End of Iteration $j$}\;
\nl $j \gets j+1$\;
}
\end{minipage}
\hfill 
 \begin{minipage}{.43\textwidth}
 \setcounter{AlgoLine}{0}

\CommentSty{\bf BOB'S PROTOCOL}
\BlankLine
\CommentSty{\bf // Iteration 0}\;
\nl Run Bob's protocol from Alg 1 \;
\nl \uIf{not terminated}{
\nl transmit random bits until channel step $\bdAlgEnd$\;
}
\CommentSty{\bf // End of Iteration 0}\;
\nl $j \gets 1$\; 
\nl \While{ still present}{
\CommentSty{\bf // Iteration $j$}\;
\nl $\fplen_j \gets \beta(j + \log L)$\;
\nl $\rep_j \gets 2^j \lceil\frac{\fplen_j}{\fplen}\rceil \wedge \fplen_j$\;
\nl $\nrj \gets 2^{j-1} \lceil 8L/\fplen \rceil$\;
\nl Run Bob's protocol from Algorithm 2, with parameters $\nrj, \fplen_j, \rep_j$\;
\CommentSty{\bf // End of Iteration $j$}\;
\nl $j \gets j+1$\;
}

\end{minipage}
\end{algorithm*}

Let $\inr := \lceil 8 L/F\rceil $ be the number of rounds in Iteration 1. Let $\nrj := 2^{j-1} \inr$ be the number of rounds in Iteration $j>1$.  Let 
$\fplen_j = 2\beta j + \fplen$ be the size of the fingerprints in Iteration $j$, where $\beta$ is the constant from the Naor and Naor hash function. Thus the hash collision probability of a single fingerprint is $2^{-2j} L^{-2}$. Each round of the iteration begins with Alice sending Bob a (1/3)-error-corrected, AMD-encoded synchronization message of length $c\fplen_j$, followed by simulation of the protocol for $\fplen_j$ channel steps, followed by Bob sending Alice a (1/3)-error-corrected, AMD-encoded fingerprint of length $c\fplen_j$. Here $c$ is the constant factor blowup we get from the ECC and AMD encodings, but for technical reasons we will further ensure that it is at least  5.  Thus, the total round length is $(2c+1)\fplen_j  \ge 11 \fplen_j$.  We will let $\twocplus$ equal $(2c+1)$.  

As in Algorithm 1, Alice will decide whether to update her verified transcript and advance to the next block of $\pi$ or to rewind to redo the current block, based on whether she receives a fingerprint from Bob that matches the fingerprint of her own transcript. Similarly, Bob will decide whether to join in the simulation of $\pi$ or to transmit random bits until the end of the round based on receiving or failing to receive Alice's synchronization message at the round's start. Where the round differs from a round in Algorithm 1, is in the actual simulation of $\pi$. For the whole iteration, a fixed number of bits of $\pi$ will be simulated per round. Each bit will be repeated $\rep_j =   2^{j-1}\lceil\fplen_j/\fplen \rceil \wedge \fplen_j$ times. \footnote{We remind the reader that $x\wedge y$ denotes the minimum of $x$ and $y$, while $x\vee y$ denotes their maximum.} 
The receiving party will use majority filtering to infer the transmitted bit. Since $\fplen_j$ time steps in the round are allocated to protocol simulation, this allows $\lfloor \fplen_j/\rep_j\rfloor$ bits of 
$\pi$ to be simulated.

Notice that the number of rounds doubles from one iteration to the next. Also, the number of repetitions of each simulated bit also roughly doubles between iterations, at least until it hits its cap, which is a constant fraction of the length of the round. This is the so-called doubling trick, (though in our case perhaps it should be quadrupling) which results in the overall cost being dominated by the cost in the last (or second to last) iteration.

%
%
%

\section{Unbounded $T$ - Analysis} \label{sec:analysis-unbounded}

We now analyze the main algorithm presented in Section~\ref{sec:alg-unbounded}.  As in Section~\ref{sec:bounded-analysis}, we begin by noting that a hash collision or an AMD code failure will cause the algorithm to fail. Additionally, the algorithm could fail during the padding rounds, if the adversary happens to flip bits in such a way as to cause Alice's random bits to look like silence, resulting in Bob's premature departure.   

In Section~\ref{sec:badpr} we will show that with high probability each of these events does not occur. Meanwhile, throughout this section we will assume without further mention that we are in the good event where none of the undesirable events occur. 

\subsection{Alice and Bob are both present}

\begin{lemma}
For every $j \ge 1$, Alice and Bob are always synchronized.  That is, they begin the iteration as well as every round therein at the same time. 
\end{lemma}
\begin{proof}
Alice and Bob synchronize themselves after Iteration 0 by both starting Iteration 1 at channel step 
$\bdAlgEnd +1$. Thereafter, for each $j\ge 1$,  they have the same round sizes $\twocplus \fplen_j$ and number of rounds 
$\nrj$ in Iteration $j$, so that they remain synchronized. 	
\end{proof}

We will call a round \emph{corrupted} if enough bits are flipped in the round that the 
bits of $\pi$ being simulated cannot be recovered or verified by Alice. We will call it \emph{uncorrupted} or \emph{progressive} if it is not corrupted in the above sense. 

\begin{lemma} Each round is either corrupted at a cost of at least $\rep_j/2$ to the adversary or results in $\lfloor \fplen_j/\rep_j \rfloor$ bits of progress in $\pi$.
\end{lemma}

\begin{proof} Since each simulated protocol bit is sent $\rep_j$ times, with majority filtering at the receiving end, it costs the adversary $\rep_j/2$ to  corrupt the repetition-encoded bit. It costs the adversary at least $c \fplen_j /3 \ge \rep_j/2$ to corrupt Alice's synchronization message or Bob's fingerprint since these are protected by error-correction. Thus it costs the adversary at least    $\rep_j/2$ to corrupt the round.  Otherwise, since there are $\fplen_j$ steps allocated to sending protocol bits, and each one is repeated $\rep_j$ times, the protocol is successfully simulated for  $\lfloor \frac{\fplen_j}{\rep_j}\rfloor$ bits.
\end{proof}

The following lemma is the equivalent of Lemmas~\ref{lem:termAlice} to~\ref{l:correctness} for Iteration $j$. Its proof is nearly identical to the proofs in Section~\ref{sec:termination} (indeed, it is simpler, since Iteration $j$ does not have the synchronization problems faced by Algorithm 1)  and we omit it.

\begin{lemma} Iteration $j$ has the following properties:
\begin{enumerate}
\item 	It is always the case that $\AlicesVerTrans \prefix \pi$, where $\pi$ is the padded transcript.
\item  At the beginning and end of each round, 
\[
\BobsVerTrans \prefix \AlicesVerTrans 
= \AlicesTmpTrans
\prefix \BobsTmpTrans;
\] 
where at most one of the inequalities is strict.
Moreover, at the end of a channel step in which Bob receives $\AlicesMsg$ correctly, \[ \BobsVerTrans = \BobsTmpTrans = \AlicesVerTrans. \]
\item Bob leaves after Alice.  When Alice leaves, $|\BobsVerTrans| \geq L$.
\item 	When either party terminates, their output is correct.
\end{enumerate}
\end{lemma}

\begin{lemma}\label{lem:uncorrupted}
There are at most $\nrj/4$ uncorrupted rounds in Iteration $j$
\end{lemma}
\begin{proof}
Since each uncorrupted round results in $\lfloor \fplen_j/ \rep_j \rfloor$ bits of progress in $\pi$, 
$\lceil L \rep_j / \fplen_j \rceil $ rounds are sufficient for Alice's transcript length to exceed $L$. 
One additional uncorrupted round is sufficient for Bob to catch up to Alice if necessary, using her synchronization message, and for Alice to infer from Bob's fingerprint that Bob's transcript length has exceeded $L$,  resulting in Alice's departure. After that, if a round is uncorrupted, then Bob will perceive silence on the channel, resulting in Bob's departure. Thus $\lceil L \rep_j / \fplen_j \rceil +2$ uncorrupted rounds are enough for both parties to terminate. Finally note that for all $j\ge 1$, 
\[
\frac{\rep_j}{\fplen_j}  \le \frac{2^{j-1}}{\fplen} \wedge 1  \le \frac{2^{j-1}}{\fplen}
\]
It follows that (for sufficiently large $L$) there are at most $2^j L/\fplen = \nrj/4$ uncorrupted rounds in Iteration $j$.
\end{proof}

The following corollary is immediate.
\begin{corollary}\label{cor:corr}
If $j$ is not the last iteration, then at least $3/4$ of the rounds are corrupted.
\end{corollary}

Although the adversary can flip any number of bits in a round, we will only charge him the minimum number of bit-flips required for the outcome we see in the round, \ie,  we will charge him 0 for uncorrupted rounds and $\rep_j/2$ for corrupted rounds.
Let $T_j$ denote the number of corruptions charged to the adversary in Iteration $j$. Clearly, for $j>0$
\begin{equation}\label{eqn:maxTj}
T_j \le \frac12 \nrj \rep_j  
\end{equation}
Also, we know from Section~\ref{sec:alg-bounded} that if the algorithm does not end in Iteration 0, then $T_0 \ge L/8F$. In this case, we will generously only charge the adversary that amount. In other words, if Iteration 1 is reached, either by both Alice and Bob, or by Bob alone, $T_0= \lceil L/8F \rceil$.

\begin{lemma}\label{lem:minTj}
If $j$ is not the last iteration then $T_j \ge \frac38 \nrj \rep_j$
\end{lemma}
\begin{proof}
This follows from Corollary~\ref{cor:corr}, since it costs the adversary at least $\rep_j /2 $ to corrupt a round. 
\end{proof}

\begin{lemma}
If $j$ is not the last iteration then 
\[
3T_{j-1}/2 \le T_j \le 64 T_{j-1}
\] 
\end{lemma}
\begin{proof}
If $j=1$ 
\[
T_1 \ge \frac38 \inr \rep_1 \ge \frac{3L}{\fplen} \ge 24 T_0 > 3T_0
\]
and 
\[
T_1 \le \inr \rep_1 /2 \le \frac{8L}{\fplen} =64 T_0\, .
\]
If $j>1$, then 
by \eqref{eqn:maxTj} and Lemma~\ref{lem:minTj}, 
\[
 \frac32 \frac{3\nrj \rep_j/8}{\nr_{j-1}\rep_{j-1} /2}  \le \frac{T_j}{T_{j-1}}  \le  \frac{\nrj \rep{j} /2}{3\nr_{j-1} \rep_{j-1}/8 } \le \ 64 
\]
since $\nr_{j-1} = \nrj/2$ and $\rep_{j-1} \le \rep_j \le  4\rep_{j-1}$. 
\end{proof}


\begin{lemma} \label{lem:uncorrsmall}The cost to either player due to uncorrupted rounds in Iteration $j\le \log F$ is at most 
\[
7 \twocplus \sqrt{ LT_{j-1} \fplen }   \] 
\end{lemma}
\begin{proof}
Each uncorrupted round costs the players $\twocplus\fplen_j$.  
Since there are at most $\nrj/4$ uncorrupted rounds, the resulting cost is no more than $\frac{\twocplus}{4} \nrj \fplen_j$. Since $j\le \log F$, $\rep_j = 2^{j-1}\lceil\fplen_j /\fplen\rceil$ and $\fplen_j \le 2\fplen$. Combining these we have 
\[
\fplen_j \le \fplen \sqrt{2^{2-j}\rep_j}
\]
so that
 \begin{align*} 
\frac{\twocplus}{4} \nrj \fplen_j &\le \twocplus \nr_{j-1}\fplen_{j-1} \\
&\le \twocplus \nr_{j-1} \fplen \sqrt{2^{3-j}\rep_{j-1}}\\
&\le  \twocplus \fplen \sqrt{\nr_{j-1}2^{3-j}}\sqrt{\nr_{j-1}\rep_{j-1}} \\
 & \le \twocplus\fplen \sqrt{2 \inr }\sqrt{8 T_{j}/3}\\
 &\le  \twocplus\sqrt{128L T_{j} \fplen/3}\\
&\le 7\twocplus\sqrt{LT_{j}F}  \,.  \qedhere
\end{align*}
\end{proof}

\begin{lemma}\label{lem:uncorrbig} If $j > \log F$, the cost to either player due to uncorrupted rounds in Iteration $j$ is at most 
\[
3\twocplus T_{j-1}   \] 
\end{lemma}
\begin{proof}
When $j > \log \fplen $, $\fplen_j = \rep_j$ and by Lemma~\ref{lem:minTj},
\[
 \frac{\twocplus}{4} \nrj \fplen_j = \frac{\twocplus}{4} \nrj \rep_j \le \twocplus \nr_{j-1} \rep_{j-1} \le 
 \frac{8\twocplus}{3} T_{j-1} \le 3\twocplus T_{j-1}\, .   \qedhere
\]
\end{proof}

\begin{lemma}\label{lem:corrcost}
The cost to  the players from corrupted rounds in Iteration $j$ is at most
$4 \twocplus \sqrt{2 LT_j\fplen} $ if $j\le \log F$ and  $ 2\twocplus T_j$ otherwise.
\end{lemma}
\begin{proof}
Suppose there are $k$ corrupted rounds. Then the cost to the players is $k\twocplus \fplen_j$, while the adversary's cost is $k\rep_j/2$. 
If $j \ge \log \fplen +1$, $\fplen_j = \rep_j$ and we easily see that the players' cost is at most $2\twocplus T$.
When $j\le \log \fplen$, since $k \le \nrj$,  
\begin{align*}
k\twocplus \fplen_j  &= \twocplus\sqrt{k \rep_j \fplen 2^{1-j}}\sqrt{\nrj \fplen_j}\\
&\le \twocplus\sqrt{T_j \fplen 2^{2-j}}\sqrt{2^{j}\inr \fplen}\\
&\le 2\twocplus\sqrt{8 LT_j\fplen} \, .     \qedhere
\end{align*}
\end{proof}


Collecting the various costs and noting that $T_j \le 64 T_{j-1}$, we see that for a suitably large constant $\gamma$, we have

\begin{lemma}\label{lem:totalcost}
The total cost to  the players from Iteration $j$ is at most
$\gamma \sqrt{ LT_{j-1}\log L} $ if $j\le \log F$ and $ \gamma T_{j-1}$ otherwise.
\end{lemma}

\subsection{Bob plays alone}

After Alice's verified transcript has length at least $L$, in each subsequent round, she transmits her synchronization message, and then random bits to indicate her continued presence. Once Alice has left, there is silence on the channel. To corrupt this silence, the adversary must make it look like a corrupted synchronization message followed by  random bits. Since a random string of length $\fplen_j$ has, on average, $\fplen_j/2$ alternations of bits, Bob considers the string to represent silence if it has fewer than $\fplen_j/3$ alternations. Thus, to corrupt such a round the adversary must pay at least $\fplen_j/3$.

Alice leaves when she has received word that Bob has a verified transcript of length at least $L$, and a single extra uncorrupted round thereafter  will cause Bob to leave as well.  Thus, if iteration $j$ was not Bob's last one, the adversary must have corrupted every round. If $1\le k<\nrj$ rounds are corrupted, 
Bob pays at most $(k +1)\twocplus \fplen_j \le 2k\twocplus \fplen_j $ and the adversary pays $k\fplen_j/3$. 
If $k=0$, we will generously account for the lone uncorrupted round from Iteration $j$  in Iteration $j-1$ by noting that 
$\twocplus( \nr_{j-1}  \fplen_{j-1} + \fplen_j) \le 2\twocplus( \nr_{j-1}  \fplen_{j-1})$
Finally a calculation identical to that in Lemma~\ref{lem:corrcost} shows that  Bob's cost for an iteration $j$ that he played alone is no more than 
\[
\gamma \sqrt{ LT_{j-1}\log L}      
\] 
if $j<\log F$ and 
\[ \gamma T_{j-1} \]
otherwise.

\subsection{Failure Probabilities} \label{sec:badpr}

In this section we bound the probabilities of the events that cause the algorithm to fail.

\begin{lemma} With high probability in $L$, there is no hash collision during Iteration $j$.
\end{lemma}
\begin{proof}  The fingerprint size has been selected large enough that the probability of a hash collision for a single hash  is $\frac{1}{2^{2j} L^2}$. Since there are $\nrj = 2^{j+2} L/F$ rounds in Iteration $j$, by a union bound, the probability of a hash collision during the iteration is $O\left(\frac{1}{2^j L\log L}\right)$.
\end{proof}

\begin{lemma}
With high probability in $L$,  any bit flipping of an AMD encoded message during Iteration $j$ is detected.
\end{lemma}
\begin{proof} The size of the AMD encoding has been selected so that the probability of a failure to detect a single instance of tampering is $\frac{1}{2^{2j} L^2}$. Since there are two AMD encodings per round and $2^{j+2} L/F$ rounds, again the probability that such a failure occurs during the iteration is $O\left(\frac{1}{2^jL\log L}\right)$.
\end{proof}

\begin{lemma} With high probability in $L$, Alice leaves before Bob.
\end{lemma}
\begin{proof}Bob does not terminate until he thinks Alice has left, and he does not even start checking for whether she seems to have left until after his transcript has length at least $L$. Since Bob's transcript lags behind that of Alice, this means that by the time Bob is checking for whether Alice has left, Alice either really has left, in which case it is fine for Bob to leave,  or she is transmitting  i.i.d. random bits in batches of length $\fplen_j$, between fingerprints. Since the adversary cannot see the bits, any bit flips on his part do not alter the fact that the string  received by Bob is a uniformly random bit string of length $\fplen_j$. Such a string has $\fplen_j/2$ alternations (consecutive bits that differ) in expectation. Bob leaves if he sees fewer than $\fplen_j/3$ alternations. If the string is random, the likelihood of Bob seeing fewer than $\fplen_j/3$ alternations is, by Chernoff's bound, at most $\eee^{-\fplen_j/18} \le \frac{1}{2^{2j} L^2}$ provided $\beta = \frac{\fplen_j}{2j +\log L}$ was chosen suitably large. Since there are at most $\nrj$ chances in Iteration $j$ for the adversary to try this attack, a union bound again shows that Bob leaves after Alice, except with probability   $O\left(\frac{1}{2^j L\log L}\right)$.
\end{proof}

\subsection{Putting everything together}

We will now prove our main theorem by putting all these costs together and calculating the total cost to either player and the failure probability of the algorithm.  As before, $T$ denotes the number of bits flipped by the adversary.

\begin{theorem}
  The algorithm succeeds with probability at least $1- 1/L\log L$. If it succeeds, then each player's cost is at most 
\[
L + O(\sqrt{LT \log L} + T) 
\]
\end{theorem}
\begin{proof}
First we note that for each $j\ge 0$ (Iteration 0 being Algorithm 1), the probability that Algorithm 3 fails during iteration $j$ is at most $O\left(\frac{1}{2^{2j} L \log L}\right) $. Thus the overall probability that it fails at all is
\[
O\left(\sum_{j=0}^{\infty} \frac{1}{2^{j} L \log L }\right) = O\left(\frac{1}{L \log L}\right)
\]
Thus, with high probability the algorithm succeeds.

Let $J$ denote the last iteration in which the player participates. 
If $J=0$ then Lemma~\ref{lem:A1cost} already proves that the players' total cost is at most $L + O(\sqrt{L(T+1) \log L})$.  Suppose $J\ge 1$.
For each $j$, let $\Cost(j)$ denote the player's cost from Iteration $j$. We know that 
\begin{itemize}
\item[$\bullet$] $\Cost(0) = \bdAlgEnd  \le L + \gamma \sqrt{LT_0 \log L}$ where $T_0 = L/(8\fplen)$ 
\item[$\bullet$] $\Cost(j) \le \gamma \sqrt{ L T_{j-1} \log L}$  if $1\le j \le \log F$
\item[$\bullet$] $\Cost(j) \le \gamma T_{j-1} $  if $j > \log F$
\end{itemize}
When $J \le \log F$, the player's total cost is 
\begin{align*}
\sum_{j=0}^J \Cost(j)  &\le     \Cost(0) + \sum_{j=1}^J \Cost(j)\\
&\le L+ \gamma \sqrt{L T_0 \log L} + \sum_{j=1}^J \gamma \sqrt{L T_{j-1} \log L}\\
&\le L+ \gamma \sqrt{L \log L} \left( \sqrt{(2/3)^{J-1} T_{J-1}} + \sum_{j=1}^J \sqrt{(2/3)^{J-1-j } T_{J-1}}\right)\\
&\le L+ \gamma \sqrt{L T_{J-1}\log L}  \left(  \sqrt{(2/3)^{J-1}} + \sum_{j=0}^{J-2} \sqrt{(2/3)^{j} }\right)\\
&\le L + \frac{\sqrt3 \gamma}{\sqrt3 -\sqrt2} \sqrt{L T_{J-1}\log L} \\
&= L+ \gamma' \sqrt{L T_{J-1}\log L}\\
&\le L+ \gamma' \sqrt{L T\log L}
\end{align*}
On the other hand, $T_{\lfloor \log F\rfloor} = \Theta(\nr_{\lfloor \log F\rfloor} \rep_{\lfloor \log F\rfloor}) = \Theta(L\log L)$, so that $\sqrt{LT_{\lfloor \log F\rfloor} \log L} = \Theta(T_{\lfloor \log F\rfloor})$ and for $J > \log F$ we have 
\begin{align*}
\sum_{j=0}^J \Cost(j)  &\le     \Cost(0) + \sum_{j=1}^{\lfloor \log F\rfloor} \Cost(j) + \sum_{j={\lfloor \log F\rfloor}+1}^J \Cost(j)\\
&\le L+ \gamma' \sqrt{L T_{\lfloor \log F\rfloor} \log L} + \sum_{j={\lfloor \log F\rfloor}+1}^J \gamma T_{j-1} \\
&\le L+\gamma'' T_{\lfloor \log F\rfloor} + \sum_{j={\lfloor \log F\rfloor}+1}^J \gamma T_{j-1}\\
&\le L+ O(T)
\end{align*}
Thus the players' cost is always $L+O\left(\sqrt{L(T+1)\log L} + T\right)$.
\end{proof}


\section{Some Additional Remarks}\label{sec:remarks}

\subsection*{Need for Private Channels} 

The following theorem justifies our assumption of private channels.  

\begin{theorem} \label{t:privateIsNecessary}
Consider any algorithm for interactive communication over a public channel that works with unknown $T$ and always terminates in the noise-free case.  Any such algorithm succeeds with probability at most $1/2$.
\end{theorem}

\begin{proof}
The adversary chooses some protocol $\pi$ with transcript length $L$ and some separate ``corrupted'' protocol $\pi_c$ such that 1) $\pi_C$ has transcript length $L$ and 2) Bob's individual input for $\pi_c$ is equivalent to his individual input for $\pi$.  The goal of the adversary will be to convince Bob that $\pi_c$ is the protocol, rather than $\pi$.  Note that we can always choose some appropriate pair $\pi$ and $\pi_c$ meeting the above criteria.  

Assume that if $\pi_c$ is the protocol and there is no noise on the channel, then Bob will output $\pi_c$ with probability at least $1/2$; if not, then the theorem is trivially true.  Then, the adversary sets $\pi$ to be the input protocol.  Next, the adversary simulates Alice in the case where her input protocol is $\pi_c$, and sets the bits received by Bob to be the bits that would be sent by Alice in such a case.  

Since the the algorithm eventually terminates, Bob will halt after some finite number of rounds, $X$. Using the above strategy, Bob will incorrectly output $\pi_c$ with probability at least $1/2$ and the value of $T$ will be no more than $X$.

Note that in the above, we critically rely on the fact that $T$ is unknown to Bob. 
\end{proof}


\subsection*{Communication Rate Comparison.}  

In Haeupler's algorithm~\cite{haeupler2014interactive}, the noise rate $\epsilon$ is known in advance and is used to design an algorithm with a communication rate of $1 - O(\sqrt{\epsilon \log \log 1/\epsilon})$.  Let $L'$ be the length of $\pi'$.  Then in his algorithm, $L' = O(L)$, and so the adversary is restricted to flipping $\epsilon L' = O(L)$ bits.  Thus, in his model, $T$ and $L'$ are always $O(L)$.  In our model, the values of $T$ and $L'$ are not known in advance, and so both $T$ and $L'$ may be asymptotically larger than $L$.

How do our results compare with~\cite{haeupler2014interactive}?  As noted above, a direct comparison is only possible when $T = O(L)$.  Restating our algorithm in terms of $\epsilon$, we have the following theorem.
\begin{theorem}\label{thm:comm-rate}
If the adversary flips $O(L)$ bits and the noise rate is $\epsilon$ then our algorithm guarantees a communication rate of $1-O\left(\sqrt{\frac{\log L}{L}} + \sqrt{\epsilon \log L}\right)$.  
\end{theorem}

%
\begin{proof}
When $T < L$ we also have $T < \sqrt{L(T+1)\log L}$ and our algorithm guarantees that for some $\gamma>0$, 
 \[ 
 L'= L +\gamma\sqrt{L(T+1)\log L}  
 \] 
Let  $\epsilon = T/L'$ and $R=L/L'$ be the effective noise and communication rates respectively. Then, 
\begin{align*}
R = \frac{L}{L'} &= 1 - \frac{L' -L}{L'} \\
&\ge 1- \frac{\gamma\sqrt{L(T+1)\log{L}}}{L'}  \\
&\ge 1- \gamma\frac{\sqrt{L\log L} + \sqrt{LT\log L}}{L'}\\ 
&\ge 1- \gamma\left(\frac{\sqrt{R\log L}}{\sqrt{L'}}+ \sqrt{R\epsilon\log L}\right)\\
&\ge 1 - \gamma\sqrt{\log L}\left( \frac{1}{\sqrt{L} }+ \sqrt{\epsilon} \right), 
\end{align*}
where the last line follows because $1/\sqrt{L'} \le 1/\sqrt{L}$ and $R\le 1$.
\end{proof}

We note that the additive term $\sqrt{\frac{\log L}{L}}$ arises from the fact that because we do not know the error rate ahead of time, we cannot get a communication rate of 1 even when the effective error rate turns out to be zero.


\subsection*{A Note on Fingerprint Size.}\label{sec:fpnote} 

A natural question is whether more powerful probabilistic techniques than union bound could enable us to use smaller fingerprints as done in~\cite{haeupler2014interactive}.  The variability of block sizes poses a challenge to this approach since Alice and Bob must either agree on the current block size, or be able to recover from a disagreement by having Bob stay in the listening loop so he can receive Alice's message.  If their transcripts diverge by more than a constant number of blocks, it may be difficult to make such a recovery, and therefore it seems challenging to modify our algorithm to use smaller fingerprints. However, it is a direction for further investigation.



\subsection*{A Lower Bound}

In this section, we prove a lower bound that demonstrates the near optimality of our upper bound by assuming the following conjecture by Haeupler holds~\cite{haeupler2014interactive}.  We now restate Haeupler's conjecture.\smallskip

\noindent{\bf Conjecture 1.}~{\it (Haeupler~\cite{haeupler2014interactive}, 2014) The maximal rate achievable by an interactive coding scheme for any binary error channel with random or oblivious errors is $1-\Theta(\sqrt{\epsilon})$ for a noise rate $\epsilon \rightarrow 0$. This also holds for for fully adversarial binary error channels if the adversary is computationally bounded or if parties have access to shared randomness that is unknown to the channel.}\smallskip
 
\noindent For the remainder of this section, we \textit{\textbf{assume that Haeupler's conjecture holds}} for any algorithm that succeed with high probability in L with an expected cost of at most $L'$ under   adversarial noise. For ease of exposition, we omit such statements in all of our claims below. By {\it robust} interactive communication, we mean interactive communication tolerates $T$ errors. 

We begin by showing the near optimality with respect to the communication rate achieved:
                                                                                                                                            
\begin{theorem}\label{thm:Lprime} 
Any algorithm for robust interactive communication must have $ L' = L + \Omega\left(T + \sqrt {L T}\right)$ for some $T\geq 1$.
\vspace{-3pt}
\end{theorem}
\begin{proof}
Let $T\geq 1$ be any value such that $T/L' = o(1)$.  Then, Haeupler's Conjecture applies and the expected total number of bits sent is $L'  \geq  L/(1-d\sqrt{\epsilon})$ for some constant $d>0$.  Noting that $1/(1-d\sqrt{\epsilon}) \geq 1 + d\sqrt{\epsilon}$  by the well-known sum of a geometric series, this implies that $L'  \geq  L/(1-d\sqrt{\epsilon}) \geq  (1 + d \sqrt{\epsilon})L = (1 + d \sqrt{T/L'})L$ since $\epsilon = T/L'$. 

This implies that $L/L' \leq 1/(1 + d \sqrt{T/L'})$. Now observe that $1/(1+x) = 1/(1-(-x)) \leq 1 -x +x^2$ for $|x| < 1$, again by the sum of a geometric series. Plugging in $d\sqrt{T/L'}$ for $x$, we have $1/(1 + d \sqrt{T/L'}) \leq  1 - d\sqrt{T/L'} + d^2(T/L')$. Therefore, $L/L' \leq 1 - d\sqrt{T/L'} + d^2(T/L') = 1 - d\sqrt{T/L'}(1 - d\sqrt{T/L'}) \leq 1 - d'\sqrt{T/L'}$ for some $d'>0$ depending only on $d$.

We then derive: $L  \leq  L'(1 - d'\sqrt{T/L'}) =  L' - d'\sqrt{L' T}$.  It follows that $L' \geq L + d'\sqrt{L' T} = L + \Omega(\sqrt{L T})$ since $L' \geq L$. 

Finally, we show that $\sqrt{LT} = \Theta( T + \sqrt{LT})$.  Assume that given any algorithm A for
interactive computation, we create a new algorithm A' that has
expected value of $L' = O(L)$.  To do this, A' checks based on
$\epsilon$ and $L$ whether or not Haeupler's algorithm~\cite{haeupler2014interactive} will send fewer bits in
expectation than A.  If so it runs Haeupler's algorithm.  Note that the
expected number of bits sent by A' is no more than the expected number
of bits sent by A.

Note that $T = \epsilon L'$ and for algorithm A',  the expected value
of $L' = O(L)$.  This implies that implies that $T = \epsilon O(L)$ or
$T = O(L)$.  Since $T< L$, it holds that $\sqrt{LT} = \Theta( T +
\sqrt{LT})$ which completes the proof. 
\end{proof}


\section{Conclusion} \label{sec:conc}

We have described the first algorithm for interactive communication that tolerates an unknown but finite amount of noise.  Against an adversary that flips $T$ bits, our algorithm sends $L + O\left(\sqrt{L(T+1)\log L} +T\right)$ bits in expectation where $L$ is the transcript length of the computation.  We prove this is optimal up to logarithmic factors, assuming a conjectured lower bound by Haeupler.  Our algorithm critically relies on the assumption of a private channel, an assumption that we show is necessary in order to tolerate an unknown noise rate.  

Several open problems remain including the following.  First, can we adapt our results to interactive communication that involves more than two parties?  Second, can we more efficiently handle an unknown amount of stochastic noise?  Finally, for any algorithm, what are the optimal tradeoffs between the overhead incurred when $T=0$ and the overhead incurred for $T>0$?

\subsection*{Acknowledgments}

We are grateful to Nico D{\"o}ttling, Bernhard Haeupler, Mahdi Zamani, and the anonymous reviewers for their useful discussions and comments.



\bibliography{jam}

\end{document}